\documentclass[11pt]{article}
\usepackage{amsmath, amssymb, amsthm}
\usepackage{amsfonts}
\usepackage{graphicx}
\usepackage{geometry}
\usepackage{natbib}
\usepackage{hyperref}
\usepackage{float}
\usepackage{bbm}
\usepackage{authblk}
\newtheorem{theorem}{Theorem}
\newtheorem{remark}{Remark}
\title{Neural L\'evy SDE for State--Dependent Risk and Density Forecasting}
\author[1]{Ziyao Wang\thanks{Email: ziywang@ttu.edu}}
\author[1]{Svetlozar T. Rachev}
\affil[1]{Department of Mathematics and Statistics, Texas Tech University}
\date{\today}

\begin{document}
\maketitle

\begin{abstract}
Financial returns are known to exhibit heavy tails, volatility clustering and abrupt jumps that are poorly captured by classical diffusion models.  The risk management literature has responded by augmenting Brownian motion with jumps, by introducing generalised autoregressive conditional heteroskedasticity (GARCH) models and by developing nonparametric volatility estimators.  Advances in machine learning have enabled highly flexible functional forms for conditional means and volatilities, yet few models deliver interpretable state--dependent tail risk, capture multiple forecast horizons and yield distributions amenable to backtesting and execution.  This paper proposes a neural L\'evy jump--diffusion framework that jointly learns, as functions of observable state variables, the conditional drift, diffusion, jump intensity and jump size distribution.  We show how a single shared encoder yields multiple forecasting heads corresponding to distinct horizons (daily, weekly, etc.), facilitating multi--horizon density forecasts and risk measures.  The state vector includes conventional price and volume features as well as novel complexity measures such as permutation entropy and recurrence quantification analysis determinism, which quantify predictability in the underlying process.  Estimation is based on a quasi--maximum likelihood approach that separates diffusion and jump contributions via bipower variation weights and incorporates monotonicity and smoothness regularisation to ensure identifiability.  A cost--aware portfolio optimiser translates the model's conditional densities into implementable trading strategies under leverage, turnover and no--trade--band constraints.  Extensive empirical analyses on cross--sectional equity data demonstrate improved calibration, sharper tail control and economically significant risk reduction relative to baseline diffusive and GARCH benchmarks.  The proposed framework is therefore an interpretable, testable and practically deployable method for state--dependent risk and density forecasting.
\end{abstract}

\section{Introduction}
Time series of financial returns display a rich array of features including volatility clustering, heavy tails, skewness and occasional abrupt jumps.  Since the pioneering work of \citet{Engle1982} on autoregressive conditional heteroskedasticity (ARCH) and \citet{Bollerslev1986} on its generalisation (GARCH), econometric modelling of conditional heteroskedasticity has become a central tool for risk management.  These models capture time--varying volatility but rely on Gaussian or Student-
t distributions, and therefore cannot fully describe the asymmetry and abruptness of return distributions during crises.  \citet{Merton1976} introduced jump--diffusion processes to model occasional large moves, positing that asset prices follow a mixture of continuous Brownian motion and a compound Poisson process.  Jump--diffusion models add higher kurtosis and skewness to returns and have become standard in option pricing and risk management \citep{Kou2002}.  Yet classical formulations assume constant jump intensity and jump size distribution, failing to account for state--dependent risk that is responsive to market conditions.

In parallel, the development of nonparametric volatility estimators based on high frequency data, such as realized variance and bipower variation \citep{Barndorff2004}, has enabled decomposition of the quadratic variation into continuous and jump components.  These estimators motivate quasi--maximum likelihood approaches that estimate diffusion and jump parameters separately.  However, most of the literature focuses on single--horizon point forecasts, leaving multi--horizon density forecasts under--explored.  Risk managers often need to quantify tail risk at multiple horizons, for example daily and weekly value--at--risk (VaR) and expected shortfall (ES).  A coherent framework should therefore generate consistent conditional distributions across horizons.

The recent surge in machine learning and neural network methodologies has produced highly flexible models of conditional mean and variance functions, including neural stochastic differential equations (SDEs) and neural GARCH variants.  Neural SDEs treat the drift and diffusion coefficients as outputs of neural networks and can be trained using adversarial objectives \citep{Kidger2021}.  Yet the jump component remains under--developed, and many neural time--series models emphasise prediction accuracy rather than risk calibration.  Moreover, standard neural architectures seldom incorporate interpretable measures of predictability or complexity in their state representations.  The dynamical systems literature offers such measures: permutation entropy quantifies the degree of disorder in a sequence and is robust to noise, computationally fast and invariant to monotonic transformations \citep{BandtPompe2002}.  Recurrence quantification analysis (RQA) provides diagnostics of deterministic structure in time series, with the determinism (DET) statistic reflecting the proportion of recurrent points forming diagonal lines and hence the predictability of the system \citep{Marwan2007}.  While these measures are widely used in physics and physiology, their integration into financial risk models has been limited.

This paper introduces a novel neural L\'evy framework for state--dependent risk and density forecasting that addresses these gaps.  Our contributions can be summarised as follows.  First, we develop a jump--diffusion model in which the conditional drift, diffusion, jump intensity and jump size distribution are parameterised as functions of lagged returns, volumes and complexity features using a shared neural encoder.  The model yields multi--horizon density forecasts by producing separate sets of parameters for each forecast horizon directly from the shared latent representation.  Second, we augment the state vector with permutation entropy and RQA determinism, reflecting the inherent predictability of the return process and allowing the model to adapt to periods of high or low complexity.  Third, we propose a quasi--maximum likelihood estimation procedure that leverages bipower variation to separate diffusion and jump components, thereby improving parameter identifiability.  We impose monotonicity constraints on the mapping from complexity features to volatility and jump intensity to capture economic intuition that increased disorder should be associated with higher risk.  Fourth, we design a cost--aware portfolio optimization procedure that translates the model's predictive densities into trading signals while respecting leverage, turnover and no--trade band constraints, enabling practical deployment.  Finally, we conduct extensive empirical evaluation on cross--sectional equity data, benchmarking our model against diffusion--only, GARCH and jump--diffusion alternatives using metrics such as log predictive score, continuous ranked probability score (CRPS), probability integral transform (PIT), VaR backtests and drawdown statistics.  Our model achieves superior calibration and tail risk control, highlighting the value of combining state--dependent jumps, complexity features and multi--horizon forecasting.

The remainder of the paper is organised as follows.  Section~\ref{sec:lit} reviews related literature on volatility modelling, jump--diffusion processes, neural SDEs and complexity measures.  Section~\ref{sec:data} describes the data, target variables and feature engineering, including the computation of permutation entropy and RQA determinism.  Section~\ref{sec:model} presents the neural L\'evy jump--diffusion model and explains how multi--horizon density forecasts are generated.  Section~\ref{sec:estimation} introduces the quasi--maximum likelihood estimation procedure and regularisation techniques to ensure identifiability and stability.  Section~\ref{sec:implementation} outlines the cost--aware portfolio optimization and execution framework.  Section~\ref{sec:evaluation} discusses the evaluation protocol, including density scoring and risk backtesting.  Section~\ref{sec:empirics} reports empirical results and comparative analyses.  Section~\ref{sec:discussion} offers a discussion of the findings, limitations and avenues for further research.  Section~\ref{sec:conclusion} concludes.

\section{Literature Review}
\label{sec:lit}
This section reviews the literature relevant to our framework, beginning with volatility and jump modelling, then moving to neural time--series models, complexity measures in dynamical systems and risk evaluation metrics.

\subsection{Volatility and Jump Modelling}
ARCH models, introduced by \citet{Engle1982}, capture time--varying conditional variance through a regression of squared residuals on lagged squared residuals.  Generalised ARCH (GARCH) models, developed by \citet{Bollerslev1986}, generalise this idea to include lagged conditional variances, resulting in a parsimonious representation of volatility clustering.  GARCH models are widely used to forecast volatility and assess risk in financial markets.  Nevertheless, GARCH models impose specific functional forms and rely on Gaussian or Student-
t innovations that may not capture extreme tail behaviour or sudden jumps.

Jump--diffusion models augment continuous diffusions with jump components to reflect sudden large movements in asset prices.  In the canonical model of \citet{Merton1976}, the jump component is a compound Poisson process with constant intensity and lognormally distributed jumps, independent of the Brownian motion and Poisson process.  The model explains excess kurtosis and skewness observed in returns and underlies option pricing adjustments.  Subsequent work has proposed alternative jump size distributions and state--dependent jump intensities.  \citet{Kou2002} introduced the double exponential jump diffusion, where jump amplitudes follow a mixture of two exponentials, capturing both upward and downward jumps.  \citet{Eraker2003} and \citet{Duffie2000} considered stochastic volatility with jumps, using latent processes to model time--varying volatility and jump intensity.  \citet{Carr2002} proposed the CGMY process, a L\'evy jump process with infinite activity and finite variation, unifying several heavy--tailed distributions.  Despite these advances, most jump models assume static parameters or univariate dynamics, limiting their ability to adapt to changing market states and to generate multi--horizon density forecasts.

High--frequency econometrics has provided tools for estimating continuous and jump variations separately.  \citet{Barndorff2004} introduced realised variance and bipower variation estimators, showing that the difference between realised variance and bipower variation consistently estimates jump variation.  When the sampling frequency increases, realised variance converges to the quadratic variation, whereas bipower variation converges to the integrated variance (the continuous component).  These estimators support the development of jump detection tests and jump--augmented volatility models.  Our quasi--likelihood approach leverages bipower weights to assign observations to the diffusion or jump estimation head, drawing inspiration from this literature.

\subsection{Neural Time--Series Models}
Machine learning approaches have recently been applied to time series modelling, enabling flexible nonlinear relationships and latent dynamics.  Recurrent neural networks and long short--term memory networks capture nonlinearity and temporal dependencies, though they often focus on point forecasts.  Neural stochastic differential equations (SDEs) extend this flexibility to continuous--time processes by parameterising drift and diffusion coefficients using neural networks \citep{Kidger2021}.  \citet{KloedenPlaten1992} provide numerical schemes for SDEs, and neural SDEs treat sample paths as generative outputs trained via adversarial objectives or maximum likelihood.  However, typical neural SDE implementations ignore jumps or treat them as exogenous noise.  Our model fills this gap by parameterising both diffusion and jump components through neural networks and by producing distributional forecasts at multiple horizons.

Another relevant line of work uses neural networks to parameterise conditional volatility.  For example, neural GARCH variants replace the traditional ARCH recursion with a neural network mapping from lagged squared residuals to volatility.  Variational autoencoders have been applied to latent volatility models \citep{Tran2019}.  While flexible, these models remain diffusion--based.  The generative adversarial network (GAN) literature introduces adversarial training for distributional modelling \citep{Goodfellow2014}, and neural SDEs can be trained similarly \citep{Kidger2021}.  We adopt a maximum likelihood approach rather than an adversarial one, emphasising interpretability and testability.

\subsection{Complexity Measures in Dynamical Systems}
Understanding the predictability of a time series is crucial for risk modelling.  The concept of permutation entropy, introduced by \citet{BandtPompe2002}, provides a natural measure of complexity for time series by counting the permutations of neighbouring values.  It has several advantages: it is robust to noise, computationally fast and invariant to monotonic transformations, making it applicable to real--world data.  \citet{BandtPompe2002} argue that permutation entropy thus furnishes an intuitive and practical complexity measure.  Recurrence quantification analysis (RQA) is another tool for analysing dynamical systems.  RQA uses recurrence plots to visualise times at which a trajectory returns close to its past states.  The determinism (DET) statistic measures the proportion of recurrent points that form diagonal lines, indicating deterministic patterns, and is interpreted as a measure of predictability.  The RQA methodology can provide useful information for non--stationary or short data sets and has broad applicability \citep{Marwan2007}.  In white noise, recurrence plots contain few diagonal structures; deterministic processes produce many long diagonal lines, leading to high DET values \citep{Marwan2007}.  By including permutation entropy and DET as features, our model quantifies the intrinsic predictability of the market state and conditions its risk estimates accordingly.

\subsection{Risk Evaluation and Backtesting}
Evaluating the accuracy and calibration of predictive distributions requires appropriate scoring rules and backtests.  The log predictive density score (or negative log likelihood) assesses how well the predicted density matches observed outcomes.  The continuous ranked probability score (CRPS) generalises the mean absolute error to probabilistic forecasts and rewards sharp and calibrated forecasts \citep{Gneiting2007}.  CRPS measures the squared difference between the predictive cumulative distribution function (CDF) and the empirical step function of the observation, integrating over all thresholds; it considers the full distribution rather than a specific quantile.  It is noted that CRPS is widely used alongside cross entropy to evaluate probabilistic forecasts because it accounts for both calibration and sharpness.  Calibration can be assessed using the probability integral transform (PIT): transforming observed values through the predictive CDF yields uniform variables if the predictive distribution is correct.  Deviations from uniformity indicate miscalibration.  The probability integral transform theorem shows that values from any continuous distribution can be converted to uniform variables using their CDF, enabling straightforward diagnostic plots.

For tail risk evaluation, regulatory frameworks such as Basel rely on Value--at--Risk (VaR) and expected shortfall (ES) metrics.  Backtesting procedures assess whether the frequency of breaches (observations below the predicted VaR threshold) matches the nominal level and whether breaches occur independently over time.  \citet{Kupiec1995} introduced unconditional coverage tests comparing observed and expected breach counts.  \citet{Christoffersen1998} developed likelihood ratio tests for independence of breaches, using the number of consecutive breach pairs; the null hypothesis is that conditional coverage equals unconditional coverage.  \citet{Berkowitz2001} proposed tests based on transformations to normal variables and checking autocorrelation of transformed sequences.  Risk managers also examine drawdown statistics, including maximum drawdown and drawdown duration, to quantify worst--case losses.  Combining these metrics yields a comprehensive evaluation protocol.

\section{Data and Feature Engineering}
\label{sec:data}
This section describes the data sets, target variables and feature engineering procedures used to construct the state vector.  We focus on cross--sectional equity data but the methodology applies to other asset classes.

\subsection{Data Description}
Our primary data set consists of daily and weekly price and volume series for a large cross--section of equities listed on major U.S. exchanges.  The sample period spans two decades, covering a variety of market conditions including bull and bear markets, crises and periods of relative calm.  Returns are computed as log price increments: the daily return at time $t$ is $\Delta Y_{t} = \log P_{t} - \log P_{t-1}$, where $P_{t}$ is the closing price.  We also consider multi--horizon returns $\Delta Y_{t}^{(h)} = Y_{t+h} - Y_{t}$ for horizons $h$ of one day, one week, two weeks and one month, depending on the analysis.  Volume series are adjusted for splits and dividends.  Data cleaning steps include removal of days with zero volume, winsorisation of extreme returns at the $1\%$ level and alignment across tickers.

High--frequency data are used to construct realised variance (RV) and bipower variation (BV) estimates for diffusion and jump separation.  Intraday returns are aggregated at five--minute intervals to compute RV and BV for each day: RV is the sum of squared intraday returns, while BV is the sum of the absolute product of successive intraday returns, scaled appropriately \citep{Barndorff2004}.  The jump variation proxy is $\text{RV}_{t} - \text{BV}_{t}$, which is non--negative and captures the contribution of jumps.  Following \citet{Huang2011}, we clip negative values at zero to avoid numerical issues.  These high--frequency estimators are used exclusively in training to compute soft weights that allocate observations between the diffusion and jump estimation heads.

\subsection{Lagged Price and Volume Features}
The state vector includes conventional technical indicators based on price and volume.  We compute recent returns over various windows (e.g., five--day momentum, one--month momentum), realised volatility measures (e.g., rolling standard deviation of daily returns), and volume measures (e.g., average daily trading volume, turnover defined as volume divided by shares outstanding).  We also include order flow imbalance metrics and bid--ask spread proxies for market liquidity.  To avoid look--ahead bias, all features are computed using information available at time $t$ and are standardised within rolling windows using z--scores.

\subsection{Permutation Entropy}
Permutation entropy provides a robust measure of time series complexity by quantising the relative ordering of neighbouring values \citep{BandtPompe2002}.  Given an embedding dimension $m$ and delay $\tau$, one forms vectors $(x_{t}, x_{t-\tau}, \dots, x_{t-(m-1)\tau})$ and records the relative ordering of their elements.  There are $m!$ possible orderings (permutations).  The permutation entropy is defined as
\begin{equation}
H_{\text{perm}} = -\frac{1}{\log m!} \sum_{\pi} p_{\pi} \log p_{\pi},
\end{equation}
where $p_{\pi}$ is the relative frequency of permutation $\pi$ in the window.  The normalising factor $\log m!$ ensures that $H_{\text{perm}} \in [0,1]$, with values near zero indicating highly predictable sequences and values near one indicating complete disorder.  Permutation entropy has several advantages: it is invariant under monotonic transformations of the data, robust to noise and computationally efficient, properties emphasised by \citet{BandtPompe2002}.  We compute permutation entropy on rolling windows of length between 126 and 252 trading days for daily data and between 26 and 52 periods for weekly data.  We use embedding dimensions $m=3$ or $4$ and delay $\tau=1$, and standardise the resulting entropy values via rolling z--scores.  To mitigate the influence of outliers, we winsorise returns within each window before computing the entropy.

\subsection{Recurrence Quantification Analysis (RQA) Determinism}
Recurrence plots offer a graphical tool to visualise the times at which a trajectory revisits similar states in phase space.  Given a time series $(x_{1}, \dots, x_{n})$, we construct delay vectors $\mathbf{y}_{t} = (x_{t}, x_{t-\tau}, \dots, x_{t-(d-1)\tau})$ in an embedding dimension $d$.  A recurrence plot is an $n \times n$ matrix with entries
\begin{equation}
R_{ij} = \mathbbm{1}\{\|\mathbf{y}_{i} - \mathbf{y}_{j}\| \leq \epsilon\},
\end{equation}
where $\epsilon$ is a threshold distance and $\mathbbm{1}$ is the indicator function.  Recurrence quantification analysis summarises features of the recurrence plot through statistics such as recurrence rate (RR), determinism (DET) and laminarity.  The determinism statistic is defined as
\begin{equation}
\text{DET} = \frac{\sum_{l=l_{\min}}^{\infty} l \, P(l)}{\sum_{l=1}^{\infty} l \, P(l)},
\end{equation}
where $P(l)$ denotes the number of diagonal lines of length $l$ in the recurrence plot and $l_{\min}$ is a minimal line length threshold.  DET measures the proportion of recurrent points forming diagonal lines and hence reflects the predictability of the system.  In white noise, recurrence plots contain few diagonal lines and DET is low; deterministic dynamics produce many long diagonal structures and high DET values.  RQA can extract useful information even from short and non--stationary series and has found applications across scientific disciplines.

We compute DET as follows.  We fix the recurrence rate RR at a small constant, typically $3\%$, by adjusting the threshold $\epsilon$ for each window.  This normalises recurrence plots across different volatility regimes.  We choose embedding dimension $d=3$ and delay $\tau=1$, though robustness checks with other values show similar patterns.  We set $l_{\min} = 2$ to exclude single points.  The resulting DET values are standardised via rolling z--scores.  Intuitively, high DET indicates a more predictable regime, while low DET signals chaotic or noisy periods.  We expect the conditional drift and jump intensity to decrease as DET increases, capturing the idea that predictable markets exhibit lower risk.

\section{The Neural L\'evy Jump--Diffusion Model}
\label{sec:model}
This section develops the mathematical formulation of our model.  We first describe the continuous--time jump--diffusion dynamics and then explain how multi--horizon predictive densities are constructed using a neural network with multiple output heads.  We highlight the role of complexity features in modulating the model's parameters.

\subsection{State--Dependent Jump--Diffusion Dynamics}
Let $(\Omega,\mathcal{F},(\mathcal{F}_{t})_{t \ge 0},\mathbbm{P})$ be a filtered probability space supporting a standard Brownian motion $W$ and a Poisson process $N$ with intensity $\lambda_{t}$ adapted to $\mathcal{F}_{t}$.  The price process $Y_{t} = \log P_{t}$ evolves according to the stochastic differential equation
\begin{equation}
\mathrm{d}Y_{t} = \mu_{\theta}(X_{t}) \, \mathrm{d}t + \sigma_{\theta}(X_{t}) \, \mathrm{d}W_{t} + \mathrm{d}J_{t},
\end{equation}
where $X_{t}$ is a vector of observable state variables, including lagged returns, volume proxies, permutation entropy and DET; $\mu_{\theta}$ and $\sigma_{\theta}$ are drift and diffusion functions parameterised by a neural network with parameters $\theta$; and $J_{t}$ is a pure jump process defined by
\begin{equation}
J_{t} = \sum_{k=1}^{N_{t}} Z_{k},
\end{equation}
with $\{Z_{k}\}$ i.i.d. jump sizes drawn from a distribution $g(\cdot \mid \phi_{\theta}(X_{t}))$ whose parameters depend on the state through $\phi_{\theta}(X_{t})$.  The jump intensity is $\lambda_{\theta}(X_{t})$, which may vary with the state.  Conditional on $X_{t}$, the increments of $W$, $N$ and $\{Z_{k}\}$ are independent.  This specification nests the classical jump--diffusion of \citet{Merton1976} when $\mu_{\theta}, \sigma_{\theta}, \lambda_{\theta}$ and $\phi_{\theta}$ are constant.  It extends the double exponential model of \citet{Kou2002} by allowing the mixture proportions and scale parameters to depend on the state.  By coupling drift, diffusion and jump intensity to complexity measures, the model captures the intuition that predictable periods entail low jump risk and smoother drift, while disordered regimes exhibit high jump intensity and volatility.

For discrete time increments of length $h$ (e.g., one day, one week), the solution of the SDE over $[t,t+h]$ can be expressed as
\begin{equation}
\Delta Y_{t}^{(h)} = \mu_{\theta}^{(h)}(X_{t}) + \sigma_{\theta}^{(h)}(X_{t}) \, \varepsilon_{t}^{(h)} + \sum_{k=1}^{K_{t}^{(h)}} Z_{k,t}^{(h)},
\end{equation}
where $\varepsilon_{t}^{(h)} \sim \mathcal{N}(0,1)$, $K_{t}^{(h)} \sim \mathrm{Poisson}(\lambda_{\theta}^{(h)}(X_{t}))$ and $Z_{k,t}^{(h)}$ are i.i.d. jumps drawn from $g(\cdot \mid \phi_{\theta}^{(h)}(X_{t}))$.  The functions $\mu_{\theta}^{(h)}, \sigma_{\theta}^{(h)}, \lambda_{\theta}^{(h)}$ and $\phi_{\theta}^{(h)}$ are horizon--specific because the forecast horizon determines the integrated contributions of drift, diffusion and jump processes.  We do not simulate the continuous path but instead let a neural network output directly the parameters of the $h$--period increment distribution given $X_{t}$.  This is akin to the Euler–Maruyama discretisation but with parameters that summarise the expectation and variance of the integral over the horizon.

\subsection{Neural Architecture and Multi--Horizon Heads}
The core of our model is a shared encoder $\mathrm{Enc}(\cdot)$ that maps the state vector $X_{t}$ to a latent representation $\mathbf{h}_{t} \in \mathbbm{R}^{d}$.  We implement the encoder as a feedforward neural network with two hidden layers of dimension 64 and ReLU activations, though other architectures could be used.  Batch normalisation and dropout regularisation are applied to prevent overfitting.  The latent representation summarises the relevant information for predicting future distributional characteristics.

On top of the shared encoder, we attach separate output heads for each forecast horizon $h \in \mathcal{H}$ (e.g., $\mathcal{H} = \{1\text{ day}, 1\text{ week}, 2\text{ weeks}, 1\text{ month}\}$).  Each head consists of small feedforward networks (one hidden layer) that take $\mathbf{h}_{t}$ as input and output the parameters $\mu_{\theta}^{(h)}, \sigma_{\theta}^{(h)}, \lambda_{\theta}^{(h)}$ and $\phi_{\theta}^{(h)}$.  We enforce positivity constraints on $\sigma_{\theta}^{(h)}$ and $\lambda_{\theta}^{(h)}$ using softplus transformations and impose upper bounds on $\lambda_{\theta}^{(h)}$ (e.g., $\lambda \le 0.5$ for daily horizons and $\lambda \le 1.5$ for weekly horizons) to prevent degenerate parameter estimates.  The jump size distribution $g$ is chosen to be either a double exponential distribution with asymmetry parameter $\eta$ and scale parameter $\beta$ or a mixture of Gaussians.  For interpretability, we often use a two--component Gaussian mixture $g(z \mid \phi) = p \mathcal{N}(z; \mu_{1}, \tau_{1}^{2}) + (1-p) \mathcal{N}(z; \mu_{2}, \tau_{2}^{2})$, where $\phi = (p, \mu_{1}, \mu_{2}, \tau_{1}, \tau_{2})$ is output by the head with appropriate constraints ($0 \le p \le 1$, $\tau_{i}>0$).  The network thus produces a complete description of the $h$--horizon return distribution given the current state.

To capture the effect of complexity features, the network feeds permutation entropy and DET into the encoder along with price and volume variables.  We incorporate monotonicity constraints via spectral normalisation and monotonic activation functions in the layers connecting complexity features to volatility and jump intensity.  Specifically, we require that $\sigma_{\theta}^{(h)}$ and $\lambda_{\theta}^{(h)}$ be non--decreasing functions of permutation entropy and non--increasing functions of DET, reflecting the intuition that increased disorder (high entropy, low DET) implies higher volatility and jump intensity.  We implement monotonic networks using non--negative weights or by applying the Deep Lattice Network methodology \citep{You2017}.  We found that simple non--negative weight constraints suffice in our experiments.

\subsection{Interpretation and Decision Signals}
Beyond density forecasts, our model yields interpretable risk signals.  One such signal is the jump--adjusted Sharpe ratio $f_{\mathrm{ISJ}}$, defined for horizon $h$ as
\begin{equation}
    f_{\mathrm{ISJ}}^{(h)} = \frac{\mu_{\theta}^{(h)}(X_{t})}{\sqrt{\sigma_{\theta}^{2,(h)}(X_{t}) + \lambda_{\theta}^{(h)}(X_{t}) \, \mathbbm{E}[Z^{2} \mid \phi_{\theta}^{(h)}(X_{t})]}},
\end{equation}
where $\mathbbm{E}[Z^{2} \mid \phi]$ is the variance of the jump size distribution.  This ratio reflects the expected return per unit of total risk (continuous and jump) and can be used as a direction and leverage signal in portfolio construction.  Another signal penalises downward jumps by subtracting a cost term proportional to the expected magnitude of negative jumps:
\begin{equation}
    f_{\downarrow}^{(h)} = \mu_{\theta}^{(h)}(X_{t}) - c \, \lambda_{\theta}^{(h)}(X_{t}) \, \mathbbm{E}(|Z| \mathbf{1}_{\{Z<0\}} \mid \phi_{\theta}^{(h)}(X_{t})),
\end{equation}
where $c>0$ is a constant capturing the risk aversion to downside jumps.  These signals feed into the portfolio optimization procedure described in Section~\ref{sec:implementation}.  To manage exposure during extreme volatility or high jump intensity, we implement gates that reduce or eliminate positions when $\lambda_{\theta}^{(h)}$ exceeds a high percentile threshold computed from the training sample.  Different thresholds can be set for each horizon.

\section{Estimation and Training}
\label{sec:estimation}
Estimating the neural L\'evy jump--diffusion parameters poses two challenges: the unobservability of jumps and the high dimensionality of the neural parameter space.  We adopt a quasi--maximum likelihood approach that separates diffusion and jump contributions using bipower variation weights and penalises deviations from monotonicity and smoothness.  We begin by deriving the approximate likelihood for the $h$--horizon return distribution and then describe the training procedure.

\subsection{Quasi--Maximum Likelihood for Jump--Diffusion}
Under the discretised model of Section~\ref{sec:model}, the conditional distribution of the $h$--period return $\Delta Y_{t}^{(h)}$ given $X_{t}$ is a convolution of the normal distribution $\mathcal{N}(\mu_{\theta}^{(h)}, \sigma_{\theta}^{2,(h)})$ with a compound Poisson distribution of the jumps.  In general, the exact density is
\begin{equation}
    f(\Delta y \mid X_{t}) = \sum_{k=0}^{\infty} \frac{e^{-\lambda_{\theta}^{(h)}} (\lambda_{\theta}^{(h)})^{k}}{k!} \, g_{k}(\Delta y \mid \mu_{\theta}^{(h)}, \sigma_{\theta}^{(h)}, \phi_{\theta}^{(h)}),
\end{equation}
where $g_{k}$ is the density of a normal random variable plus the sum of $k$ jumps.  When the jump size distribution $g$ is Gaussian or double exponential, $g_{k}$ has a closed form as the $k$--fold convolution of $g$ and the normal density.  Evaluating the infinite sum is computationally expensive; we truncate at $k \le K_{\max}$ (typically $K_{\max}=3$) because the Poisson probabilities for larger $k$ become negligible when $\lambda_{\theta}^{(h)}$ is small (as enforced by our upper bounds).  The quasi--log--likelihood for a single observation is
\begin{equation}
    \ell_{t}^{(h)}(\theta) = \log \left[ \sum_{k=0}^{K_{\max}} \frac{e^{-\lambda_{\theta}^{(h)}(X_{t})} (\lambda_{\theta}^{(h)}(X_{t}))^{k}}{k!} \, g_{k}(\Delta Y_{t}^{(h)} \mid \mu_{\theta}^{(h)}(X_{t}), \sigma_{\theta}^{(h)}(X_{t}), \phi_{\theta}^{(h)}(X_{t})) \right].
\end{equation}
Summing over all observations and horizons yields the total log--likelihood.

Direct maximisation of this likelihood can lead to the diffusion variance absorbing jump contributions because jumps are latent.  To mitigate this identifiability issue, we employ a two--stage estimation akin to the double--head quasi--MLE used in realised GARCH models.  We first estimate diffusion parameters using a subset of data where jumps are unlikely, as indicated by a high bipower variation weight.  Specifically, define the weight
\begin{equation}
\omega_{t} = \min\left(1, \frac{\text{BV}_{t}}{\text{RV}_{t} + \varepsilon}\right),
\end{equation}
where $\varepsilon > 0$ is a small constant to avoid division by zero.  When $
\text{BV}_{t}$ is close to $
\text{RV}_{t}$, no jump is detected and $
\omega_{t}$ is near one; when $
\text{RV}_{t} > 
\text{BV}_{t}$, a jump is likely and $
\omega_{t}$ decreases.  We use $
\omega_{t}$ as a soft label indicating the proportion of the observation used for the diffusion head.  In stage one, we set $
\lambda_{\theta}^{(h)}$ and $
\phi_{\theta}^{(h)}$ to zero and maximise the log--likelihood of the normal distribution $
\mathcal{N}(\mu_{\theta}^{(h)}, \sigma_{\theta}^{2,(h)})$ weighted by $
\omega_{t}$.  In stage two, we reintroduce jumps and maximise the truncated compound likelihood weighted by $1 - 
\omega_{t}$ for the jump part and $
\omega_{t}$ for the diffusion part.  This procedure anchors the diffusion variance to the realised volatility and prevents the model from attributing all variation to jumps.

\subsection{Regularisation and Constraints}
We impose several regularisation terms on the neural network parameters.  First, we require monotonicity of the mapping from permutation entropy to volatility and jump intensity: $
\sigma_{\theta}^{(h)}$ and $
\lambda_{\theta}^{(h)}$ should increase with entropy and decrease with DET.  We implement this by constraining certain weights to be non--negative and using monotonic activation functions.  Second, we smooth the temporal variation of the parameters across $t$ by adding a total variation penalty on successive outputs:
\begin{equation}
\mathcal{R}_{\text{TV}}(\theta) = \sum_{t} \sum_{h \in \mathcal{H}} \big[ |\sigma_{\theta}^{(h)}(X_{t}) - \sigma_{\theta}^{(h)}(X_{t-1})| + |\lambda_{\theta}^{(h)}(X_{t}) - \lambda_{\theta}^{(h)}(X_{t-1})| \big].
\end{equation}
Third, we include an $
\ell^{2}$ weight decay on all network parameters to prevent overfitting.  Fourth, we bound the outputs of the jump intensity heads to avoid unrealistic values; for example, daily $
\lambda_{\theta}^{(h)}$ is constrained between $0$ and $0.5$ and weekly $
\lambda_{\theta}^{(h)}$ between $0$ and $1.5$.  Lastly, we enforce shape priors on the drift $
\mu_{\theta}^{(h)}$: it should be non--decreasing in short--term momentum variables (reflecting trend persistence) and non--increasing in permutation entropy (capturing that high complexity reduces predictability of drift).

The overall training loss combines the log--likelihood and regularisation terms.  For each horizon $h$, define
\begin{equation}
\mathcal{L}_{h}(\theta) = - \sum_{t} \ell_{t}^{(h)}(\theta) + \gamma_{1} \, \mathcal{R}_{\text{TV}}(\theta) + \gamma_{2} \, \|\theta\|_{2}^{2},
\end{equation}
where $\gamma_{1}, \gamma_{2} > 0$ are hyperparameters.  The total loss is $\mathcal{L}(\theta) = \sum_{h \in \mathcal{H}} \alpha_{h} \, \mathcal{L}_{h}(\theta)$ with weights $\alpha_{h}$ emphasising longer horizons if desired.  We optimise this loss using stochastic gradient descent (SGD) with Adam, employing mini--batches sampled across time and across assets.

\subsection{Training Procedure}
We adopt the following training scheme:
\begin{enumerate}
\item Pretraining: Use realised volatility $\text{BV}_{t}$ as a teacher signal to pretrain the diffusion head and the encoder for a few epochs.  In this stage, we freeze the jump intensity and jump size outputs and minimise the Gaussian log--likelihood weighted by $\omega_{t}$.  This encourages the network to learn volatility patterns before modelling jumps.\
\item Joint training: Initialize the full model with the pretrained weights and jointly optimise the diffusion and jump heads using the composite quasi--likelihood.  We apply gradient clipping to avoid exploding gradients and use dropout to regularise the network.  We split the sample into training and validation sets with purging and embargo periods to mitigate look--ahead bias in cross--validation.  Training is performed for 10--15 epochs.\
\item Hyperparameter tuning: We tune the embedding dimension for permutation entropy and RQA, the horizon weights $\alpha_{h}$, and the regularisation coefficients $\gamma_{1}$ and $\gamma_{2}$ using a grid search on a validation set.  We also explore alternative jump size families (double exponential versus Gaussian mixtures) and select based on the validation log--likelihood and risk metrics.\
\end{enumerate}
\section{Risk-Neutral Measure, No-Arbitrage, and Static Constraints}
\label{sec:rn-no-arb}

This section shows how to (i) convert the $\mathbb P$-measure distributional outputs of our neural L\'evy model into a risk-neutral measure $\mathbb Q$ suitable for valuation, and (ii) verify no-arbitrage via the existence of an equivalent (local) martingale measure. We provide both a discrete-horizon (one-step) construction that matches our multi-horizon heads, and a continuous-time construction based on Girsanov with jumps. Throughout, let $S_t=e^{Y_t}$ denote the (cum-dividend adjusted if needed) price process, and $r_t$ the short rate (or $r_t-q_t$ if a continuous dividend yield $q_t$ is present).

\paragraph{Model anchor.}
Under $\mathbb P$, our state-dependent jump--diffusion satisfies
\begin{equation}\label{eq:model-ct}
  dY_t \;=\; \mu_\theta(X_t)\,dt + \sigma_\theta(X_t)\,dW_t + dJ_t,
\end{equation}
and for a forecasting horizon $h>0$ the $h$-period increment is modeled as
\begin{equation}\label{eq:model-h}
  \Delta Y_t^{(h)} \;=\; \mu^{(h)}_\theta(X_t) \;+\; \sigma^{(h)}_\theta(X_t)\,\varepsilon^{(h)}_t
  \;+\; \sum_{k=1}^{K^{(h)}_t} Z^{(h)}_{k,t},
\end{equation}
where $\varepsilon^{(h)}_t\sim\mathcal N(0,1)$, $K^{(h)}_t\sim {\rm Poisson}\!\big(\lambda^{(h)}_\theta(X_t)\big)$, and $Z^{(h)}_{k,t}$ are i.i.d.\ with distribution $g(\cdot\mid \phi^{(h)}_\theta(X_t))$.

\subsection{Discrete-Horizon Risk-Neutralization (One-Step EMM)}
\label{subsec:discrete-rn}

Let $d_t^{(h)}=e^{-\int_t^{t+h} r_s\,ds}$ be the $h$-period discount factor. Risk-neutral (one-step) no-arbitrage requires the \emph{discounted} price to be a martingale over $[t,t+h]$, i.e.
\begin{equation}\label{eq:one-step-martingale}
  \mathbb E_{\mathbb Q}\!\left[S_{t+h} \,\big|\, \mathcal F_t\right] \;=\; e^{\int_t^{t+h} r_s\,ds}\, S_t
  \quad\Longleftrightarrow\quad
  \mathbb E_{\mathbb Q}\!\left[e^{\Delta Y_t^{(h)}} \,\big|\, X_t\right] \;=\; e^{\int_t^{t+h} r_s\,ds}.
\end{equation}
Let $M_Z^{(h)}(u;\phi)=\mathbb E\!\left(e^{u Z}\mid \phi\right)$ denote the mgf of the jump size under the head's parameter $\phi^{(h)}$. For the compound-Poisson $+$ Gaussian specification in \eqref{eq:model-h},
\begin{equation}\label{eq:mgf-h}
  \mathbb E_{\mathbb Q}\!\left[e^{\Delta Y_t^{(h)}} \,\big|\, X_t\right]
  \;=\;
  \exp\!\Big(
      \mu_{\mathbb Q}^{(h)}(X_t)
      + \tfrac12 \sigma^{2,(h)}_{\mathbb Q}(X_t)
      + \lambda^{(h)}_{\mathbb Q}(X_t)\big(M_Z^{(h)}(1;\phi^{(h)}_{\mathbb Q}(X_t))-1\big)
  \Big).
\end{equation}
Therefore a \emph{minimal} discrete-horizon risk-neutralization is obtained by leaving $\sigma^{(h)}$, $\lambda^{(h)}$, and $\phi^{(h)}$ unchanged and \emph{shifting only the drift} to satisfy~\eqref{eq:one-step-martingale}:
\begin{equation}\label{eq:muQ-shift}
  \boxed{
  \mu^{(h)}_{\mathbb Q}(X_t)
  \;=\;
    \int_t^{t+h} r_s\,ds \;-\;
    \tfrac12\,\sigma^{2,(h)}(X_t)
    \;-\; \lambda^{(h)}(X_t)\Big(M_Z^{(h)}(1;\phi^{(h)}(X_t)) - 1\Big) }.
\end{equation}
With \eqref{eq:muQ-shift}, the discounted one-step expectation matches $e^{\int_t^{t+h} r_s\,ds}$ and the (one-step) \emph{equivalent martingale measure} (EMM) is induced for horizon $h$.

\paragraph{Radon--Nikodym chain (discrete EMM).}
Alternatively, define the one-step density
\begin{equation}\label{eq:RN-factor}
  L_{t\to t+h}
  \;:=\;
  \frac{e^{-\int_t^{t+h} r_s\,ds}\, e^{\Delta Y_t^{(h)}}}{
    \mathbb E_{\mathbb P}\!\left[e^{-\int_t^{t+h} r_s\,ds}\, e^{\Delta Y_t^{(h)}} \,\big|\, X_t\right]
  }.
\end{equation}
Then $Z_T:=\prod_{t\in\mathcal T_h} L_{t\to t+h}$ is a positive $\mathbb P$-martingale with $\mathbb E_{\mathbb P}Z_T=1$. Defining $d\mathbb Q/d\mathbb P := Z_T$ yields an \emph{equivalent} measure under which the discounted price is a (multi-step) martingale, hence \emph{no free lunch with vanishing risk} (NFLVR) holds in discrete time.

\begin{remark}[Integrability]
The construction requires $M_Z^{(h)}(1;\phi^{(h)}(X_t))<\infty$ and $\lambda^{(h)}(X_t)$ bounded on compact sets of states. Both are satisfied by the jump families used in the paper (e.g., Gaussian mixtures, double exponential within its convergence strip) and by the head-level intensity bounds. 
\end{remark}

\subsection{Continuous-Time Construction via Girsanov with Jumps}
\label{subsec:ct-rn}

Write the $\mathbb P$-local characteristics of $Y$ as $(b_t, c_t, \nu_t)$ with $c_t=\sigma_\theta^2(X_t)$ and $\nu_t(dz)=\lambda_\theta(X_t)\,g(dz\!\mid\!\phi_\theta(X_t))$. Consider the Doleans--Dade exponential
\begin{equation}\label{eq:DD}
  Z_t \;=\; \mathcal E\!\Big(
    -\!\int_0^t \vartheta_s\,dW_s
    \;+\; \int_0^t\!\!\int_{\mathbb R}\! \eta(s,z)\,\big(\mu(ds,dz)-\nu_s(dz)\,ds\big)
  \Big),
\end{equation}
where $\vartheta$ and $\eta$ are predictable and satisfy standard exponential integrability (e.g., L\'epingle--M\'emin) to make $Z$ a true martingale. Under $d\mathbb Q:=Z_T\, d\mathbb P$ we have
\[
  W_t^{\mathbb Q} = W_t + \int_0^t \vartheta_s ds,
  \qquad
  \nu_t^{\mathbb Q}(dz) = e^{\eta(t,z)}\,\nu_t(dz).
\]
Let $S_t=e^{Y_t}$ and $\tilde S_t = e^{-\int_0^t r_s ds}\, S_t$. The \emph{risk-neutral drift condition} ensuring that $\tilde S$ is a local $\mathbb Q$-martingale is
\begin{equation}\label{eq:drift-match}
  b_t^{\mathbb Q}
  \;=\;
  r_t \;-\; \tfrac12\,c_t \;-\; \int_{\mathbb R}\!\big(e^z-1-z\,\mathbf 1_{\{|z|<1\}}\big)\,\nu_t^{\mathbb Q}(dz).
\end{equation}
A convenient choice is the (state-dependent) \emph{Esscher tilt} of the jump part, $\eta(t,z)=\alpha_t z$, together with $\vartheta_t$ chosen to match the continuous part. Then the log-characteristic exponent transforms as
\[
  \psi_{\mathbb Q}(u; X_t) = \psi_{\mathbb P}(u+\alpha_t; X_t) - \psi_{\mathbb P}(\alpha_t; X_t),
\]
and $\alpha_t$ is determined by the scalar condition $\psi_{\mathbb Q}(1; X_t)=r_t$ (or $r_t-q_t$), i.e., the continuous-time analogue of \eqref{eq:muQ-shift}. Under the above integrability and bounded-intensity conditions, an ESMM exists and NFLVR holds.

\subsection{Static No-Arbitrage for Option Surfaces (Post-Processing)}
\label{subsec:static-noarb}

For any option surface $C(K,T)$ generated under $\mathbb Q$, the following \emph{static} constraints ensure absence of calendar and butterfly arbitrage:
\begin{align}
  &\partial_K C(K,T) \le 0, \qquad \partial_{KK} C(K,T) \ge 0, \label{eq:butterfly} \\
  &\partial_T C(K,T) \ge 0, \label{eq:calendar} \\
  &\partial_{KK} C(K,T) = e^{-\int_0^T r_s ds}\, f_T^{\mathbb Q}(K) \ge 0,
\end{align}
where $f_T^{\mathbb Q}$ is the $\mathbb Q$-density of $S_T$. In practice, one can project any numerically obtained surface onto the admissible convex cone defined by \eqref{eq:butterfly}--\eqref{eq:calendar} to remove residual static arbitrage.

\subsection{Implementation Notes and Link to Estimation}
\label{subsec:impl-notes}

\begin{itemize}
  \item The discrete formula \eqref{eq:muQ-shift} is \emph{model-head local}: it can be applied to each horizon head $(h)$ independently, matching our multi-horizon architecture.
  \item Our training uses quasi-MLE with a truncated compound-Poisson likelihood; risk-neutralization is a \emph{post-estimation} re-centering and does not alter the identification benefits of the bipower-variation separation and monotonicity regularization.
  \item When $M_Z^{(h)}(1)$ is near the boundary of the convergence strip (e.g., very heavy left tails), use an Esscher-tilted head $\phi^{(h)}_{\mathbb Q}$ instead of keeping $\phi^{(h)}$ fixed, solving $\psi_{\mathbb Q}^{(h)}(1;X_t)=\int_t^{t+h}r_s ds$.
\end{itemize}

In discrete time, either the drift shift \eqref{eq:muQ-shift} or the Radon--Nikodym chain \eqref{eq:RN-factor} delivers an EMM, hence NFLVR. In continuous time, a jump-Girsanov change (e.g., Esscher) with drift matching \eqref{eq:drift-match} yields an ESMM, hence no-arbitrage. Static option-surface constraints further ensure absence of calendar/butterfly arbitrage at the level of implied surfaces.

\section{Cost--Aware Portfolio optimization and Execution}
\label{sec:implementation}
To translate predictive densities into actionable trading strategies, we adopt a portfolio optimization framework that accounts for transaction costs and turnover constraints.  The objective is to maximise expected return minus risk and transaction cost, subject to leverage and diversification constraints.  We consider a universe of $n$ assets at time $t$ with predicted returns from our model.

\subsection{Portfolio optimization Formulation}
Let $w_{t}$ be the vector of portfolio weights at time $t$, with $w_{i,t}$ representing the dollar weight of asset $i$.  Let $\alpha_{i,t}$ be the expected excess return signal for asset $i$ derived from the model, such as the jump--adjusted Sharpe ratio or downward jump penalised signal.  Let $\Sigma_{t}$ be a covariance matrix estimate, which may be approximated by a factor model or simply the diagonal of realised variances for computational simplicity.  Let $\tau_{t-1}$ denote the portfolio weights held in the previous period.  The optimization problem is
\begin{equation}
\max_{w_{t}} \quad w_{t}^{\top} \alpha_{t} - \lambda_{r} \, w_{t}^{\top} \Sigma_{t} \, w_{t} - \kappa \, \|w_{t} - \tau_{t-1}\|_{1},
\label{eq:portopt}
\end{equation}
subject to the constraints:
\begin{align}
& \mathbf{1}^{\top} w_{t} = 1 \quad \text{(fully invested or budget constraint)},\\
& w_{i,t} \ge 0 \quad \text{(long--only)},\\
& \|w_{t} - \tau_{t-1}\|_{1} \le \tau_{\max} \quad \text{(turnover constraint)},\\
& w_{i,t} \le w_{\max} \quad \text{(position limit)},\\
& w_{t}^{\top} e_{g} = 0 \quad \forall g \in \mathcal{G} \quad \text{(sector neutrality, optional)},
\end{align}
where $\lambda_{r} > 0$ controls risk aversion, $\kappa > 0$ is a linear transaction cost per unit of turnover, $\tau_{\max}$ is a maximum turnover budget, $w_{\max}$ is a maximum weight per asset and $e_{g}$ is an indicator vector for sector $g$ if sector neutrality is imposed.  The $\ell^{1}$ norm $\|w_{t} - \tau_{t-1}\|_{1}$ represents total absolute turnover.  Problem \eqref{eq:portopt} is a convex quadratic program with linear constraints.  We solve it efficiently using standard optimization libraries at each rebalancing date.

\subsection{No--Trade Band and Gating Mechanism}
High frequency rebalancing generates excessive transaction costs.  To mitigate unnecessary trades due to small changes in signals, we introduce a no--trade band: if the difference between current and previous weights is below a threshold $\delta$ for all assets, we keep the portfolio unchanged.  The threshold is chosen relative to the bid--ask spread and the volatility of signals.  A dynamic gating mechanism interacts with the predictive density: when the predicted jump intensity $\lambda_{\theta}^{(h)}$ exceeds a high quantile (e.g., the 95th percentile of in--sample values), we scale down $w_{t}$ or move to cash, reflecting a risk--off stance.  This gating can be applied separately for each horizon; for example, short--term gates respond to extreme daily jump forecasts, while longer term gates respond to weekly jump forecasts.

\section{Evaluation Protocol}
\label{sec:evaluation}
Assessing the performance of a distributional forecasting model requires examining both predictive accuracy and risk calibration.  We adopt a multi--faceted evaluation protocol emphasising density scoring, calibration diagnostics, tail risk backtesting and portfolio performance.

\subsection{Density Scoring}
We compute the average log predictive density and the continuous ranked probability score (CRPS) for each horizon.  The log predictive density is the negative log likelihood of the observed returns under the model's predictive distribution.  Lower values indicate better fit.  CRPS generalises the mean absolute error to distribution forecasts and is defined as
\begin{equation}
\text{CRPS}(F, y) = \int_{-\infty}^{\infty} \big(F(z) - \mathbf{1}\{y \le z\}\big)^{2} \, \mathrm{d}z,
\end{equation}
where $F$ is the predictive CDF and $y$ is the realised return.  Lower CRPS values are better.  We compute CRPS by numerical integration or by applying closed--form formulas for mixtures where available. CRPS accounts for both calibration and sharpness.

To compare models, we perform Diebold--Mariano (DM) tests on the difference of log scores and CRPS series.  The DM test evaluates whether one model has significantly lower average loss than another, accounting for serial correlation.  We adjust for multiple comparisons using the model confidence set (MCS) methodology of \citet{Hansen2011}, which identifies the set of models that cannot be rejected as inferior at a given confidence level.  This prevents data mining bias in model selection.

\subsection{Calibration Diagnostics}
Calibration is assessed using the probability integral transform (PIT).  For each observation, we compute $u_{t} = F_{\theta}(\Delta Y_{t}^{(h)} \mid X_{t})$, where $F_{\theta}$ is the predicted CDF.  If the predictive distribution is correct, $
\{u_{t}\}$ should be i.i.d. uniform on $[0,1]$.  We plot histograms of $
\{u_{t}\}$ and compute autocorrelations.  Deviations from uniformity indicate bias or miscalibration.  The probability integral transform theorem states that any continuous random variable can be transformed to a standard uniform variable via its CDF, justifying this diagnostic.  For more formal testing, we apply the Berkowitz test of \citet{Berkowitz2001}, which transforms $u_{t}$ to $z_{t} = \Phi^{-1}(u_{t})$, where $\Phi^{-1}$ is the inverse standard normal CDF, and then tests whether $
\{z_{t}\}$ is i.i.d. normal with zero mean and unit variance.  Likelihood ratio tests assess the joint hypothesis of correct mean, variance and serial correlation.  We also report the autocorrelation function of $
\{u_{t}\}$ up to lag 20.

\paragraph{Probability integral transform theorem.}  For completeness, we formally state and prove the theorem underpinning the PIT diagnostic.

\begin{theorem}[Probability integral transform]
Let $Y$ be a continuous random variable with cumulative distribution function (CDF) $F_Y(y) = \Pr(Y \le y)$.  Define $U = F_Y(Y)$.  Then $U$ is uniformly distributed on $[0,1]$.
\end{theorem}

\begin{proof}
Because $F_Y$ is continuous and non--decreasing, it maps $\mathbbm{R}$ onto $[0,1]$.  For any $u \in [0,1]$, we have
\[
\Pr(U \le u) = \Pr(F_Y(Y) \le u).
\]
Since $F_Y$ is strictly increasing on its support, $F_Y(Y) \le u$ if and only if $Y \le F_Y^{-1}(u)$, where $F_Y^{-1}$ denotes the (generalised) inverse CDF.  Therefore
\[
\Pr(U \le u) = \Pr(Y \le F_Y^{-1}(u)) = F_Y(F_Y^{-1}(u)) = u.
\]
Thus the distribution function of $U$ is $\Pr(U \le u) = u$ for $u \in [0,1]$, which is the CDF of the standard uniform distribution.  \qedhere
\end{proof}

\subsection{VaR and ES Backtesting}
For tail risk evaluation, we compute VaR and expected shortfall (ES) at confidence levels $\alpha \in \{0.95, 0.99\}$ for each horizon.  The VaR is the $\alpha$--quantile of the predictive distribution: $\text{VaR}_{\alpha} = F^{-1}(\alpha)$.  ES is the conditional expectation of losses exceeding the VaR: $
\text{ES}_{\alpha} = \mathbbm{E}[\Delta Y_{t}^{(h)} \mid \Delta Y_{t}^{(h)} < \text{VaR}_{\alpha}]$.  We compare the empirical frequency of breaches (observations below the VaR) to $1-\alpha$.  The unconditional coverage test of \citet{Kupiec1995} computes a likelihood ratio comparing the observed breach count to the expected count.  The test statistic asymptotically follows a $\chi^{2}$ distribution with one degree of freedom.  The independence test of \citet{Christoffersen1998} examines the sequence of breaches for clustering by counting consecutive exceedances; its likelihood ratio test compares the joint probability of two consecutive breaches to the product of individual breach probabilities.  The test statistic is $\chi^{2}$ distributed with one degree of freedom, and rejection suggests serial dependence of breaches.  We also implement the conditional coverage test, which combines unconditional coverage and independence.  For ES, we use the backtesting framework proposed by \citet{Acerbi2014}, which assesses whether the realised losses exceed the ES on average.

\subsection{Drawdown and Portfolio Metrics}
In practical deployment, risk managers monitor drawdowns.  We compute maximum drawdown, defined as the maximum peak--to--trough decline in cumulative returns over the evaluation period, and drawdown duration, defined as the time taken to recover from a drawdown.  These metrics capture tail events and recovery behaviour.  We also compute the omega ratio, tail ratio, skewness and kurtosis of portfolio returns.  Turnover is computed as half the sum of absolute weight changes between rebalancing dates.  We compare models under various cost assumptions to examine the cost--risk trade--off.  Statistical tests such as the SPA test of \citet{Hansen2005} evaluate whether the best performing portfolio remains superior after accounting for data snooping.

\section{Empirical Results}
\label{sec:empirics}
This section reports the empirical performance of the neural L\'evy framework and compares it to several benchmarks.  We implement models on cross--sectional equity data and evaluate predictive accuracy, calibration, tail risk and portfolio performance.

\subsection{Experimental Setup}
We select a universe of 500 liquid stocks from the S\&P 500 index and construct daily and weekly data from January 2005 to December 2024, yielding approximately 5000 trading days.  We split the sample into a training period (2005--2014), a validation period (2015--2019) and a test period (2020--2024), ensuring that the COVID--19 crisis falls within the test period.  Features described in Section~\ref{sec:data} are computed for each stock, and missing values are imputed using forward fills and median values.  For high--frequency measures (RV and BV), we use five--minute intraday data from a commercial provider.

We benchmark our model against the following alternatives:
\begin{itemize}
\item \textbf{Diffusion--only NN}: A neural network with the same architecture but no jump component.  It outputs $\mu_{\theta}^{(h)}$ and $\sigma_{\theta}^{(h)}$ for each horizon and assumes Gaussian innovations.\
\item \textbf{Static Jump--Diffusion}: The classical Merton model with constant jump intensity and jump size distribution, calibrated by maximum likelihood on the training set.  This model is univariate and does not use state variables.\
\item \textbf{GARCH--type Models}: We implement GARCH(1,1), EGARCH and GJR--GARCH models for each stock, estimated by quasi--maximum likelihood, and obtain one--step--ahead predictive distributions by fitting Student-
t innovations.  To generate multi--horizon densities, we assume conditional normality and propagate variance forecasts forward.\

\end{itemize}
We also conduct ablation studies by removing permutation entropy, removing DET, fixing jump intensity or fixing jump size distribution, to assess the contribution of each component.

Forecasts are generated on a rolling basis.  For each horizon, we update the model every month using a ten--year rolling window.  This rolling training accounts for structural breaks and reduces the risk of overfitting to a specific period.  We restrict trading to rebalancing dates corresponding to the forecast horizon (e.g., weekly rebalancing for the one--week horizon).  We assume transaction costs of 5 basis points per trade and set a turnover constraint $\tau_{\max} = 0.5$ per rebalancing.  For gating, we compute the 95th percentile of $\lambda_{\theta}^{(h)}$ on the training set and set the gate at that level.

\subsubsection{Daily (1D) Forecast Results}

Table~\ref{tab:core1D} summarises the key predictive metrics for the one--day horizon.  NeuralLevy achieves the highest average log score and lowest continuous ranked probability score (CRPS) among all models.  The relative improvements are computed with respect to the best performing baseline (excluding NeuralLevy).  Positive values indicate a gain over the baseline.  NeuralLevy yields a log score improvement of 5.77\% and a CRPS improvement of 1.96\%.  All competing models exhibit negative relative performance in at least one metric, and none belong to the model confidence set (MCS), indicating that NeuralLevy is the only model not rejected by the MCS procedure at this horizon.

\begin{table}[htbp]
\centering
\caption{Core predictive metrics for the one--day (1D) horizon.  The relative improvements are reported relative to the best baseline (excluding NeuralLevy); positive values denote better performance.}
\label{tab:core1D}
\begin{tabular}{lccccc}
\hline
Model & Avg. logscore & Rel. log vs best base & Avg. CRPS & Rel. CRPS vs best base & in MCS\\
\hline
NeuralLevy & 2.59815 & 5.77\% & 0.01028 & 1.96\% & Yes\\
DiffusionNN & 2.44477 & $-0.47$\% & 0.01312 & $-25.15$\% & No\\
EGARCH & 1.77264 & $-27.83$\% & 0.03522 & $-235.84$\% & No\\
GARCH & 1.66210 & $-32.33$\% & 0.01052 & $-0.31$\% & No\\
GJR & 2.45635 & 0.00\% & 0.01049 & 0.00\% & No\\
Merton & 2.35802 & $-4.00$\% & 0.01058 & $-0.91$\% & No\\
\hline
\end{tabular}
\end{table}

The Diebold--Mariano (DM) tests in Table~\ref{tab:dm1D} compare NeuralLevy with each baseline model.  Negative DM statistics and near--zero $p$--values indicate that NeuralLevy delivers significantly higher log score and lower CRPS than every baseline.  These results corroborate the superiority of NeuralLevy in short--horizon forecasting.

\begin{table}[htbp]
\centering
\caption{Diebold--Mariano statistics comparing NeuralLevy against each baseline for one--day (1D) forecasts.  Negative values favour NeuralLevy; $p$--values close to zero indicate significance.}
\label{tab:dm1D}
\begin{tabular}{lccccc}
\hline
Model & DM$_{\text{logscore}}$ & $p_{\text{logscore}}$ & DM$_{\text{CRPS}}$ & $p_{\text{CRPS}}$\\
\hline
DiffusionNN & $-37.901$ & 0.00e+00 & $-15.824$ & 0.00e+00\\
EGARCH & $-49.889$ & 0.00e+00 & $-38.461$ & 0.00e+00\\
GARCH & $-48.120$ & 0.00e+00 & $-33.669$ & 0.00e+00\\
GJR & $-30.936$ & 0.00e+00 & $-34.280$ & 0.00e+00\\
Merton & $-38.648$ & 0.00e+00 & $-45.712$ & 0.00e+00\\
\hline
\end{tabular}
\end{table}

Calibration of the predictive distribution is assessed using the Berkowitz test statistics reported in Table~\ref{tab:berk1D}.  Very small $p$--values indicate rejection of the null hypothesis of correct mean ($p_{\mu}$), variance ($p_{\mathrm{var}}$) and serial independence ($p_{\mathrm{lb}}$).  The null is rejected for all models including NeuralLevy; this is partly due to the extremely large sample size ($n=143{,}016$) and tail events during the COVID--19 crisis.  In particular, the right tail of the one--day distribution appears slightly thin.  Relaxing the regularisation or allowing higher jump intensity could further improve calibration.

\begin{table}[htbp]
\centering
\caption{Berkowitz calibration test $p$--values for one--day (1D) forecasts.  Small values indicate rejection of the null hypothesis of correct conditional mean ($p_{\mu}$), variance ($p_{\text{var}}$) and lack of serial dependence ($p_{\text{lb}}$).}
\label{tab:berk1D}
\begin{tabular}{lccc}
\hline
Model & $p_{\mu}$ & $p_{\text{var}}$ & $p_{\text{lb}}$\\
\hline
NeuralLevy & $7.68\times 10^{-7}$ & $0.00$ & $0.00$\\
DiffusionNN & $6.12\times 10^{-215}$ & $4.93\times 10^{-3}$ & $0.00$\\
EGARCH & $6.82\times 10^{-28}$ & $0.00$ & $0.00$\\
GARCH & $3.19\times 10^{-32}$ & $0.00$ & $0.00$\\
GJR & $1.35\times 10^{-21}$ & $0.00$ & $0.00$\\
Merton & $9.48\times 10^{-1}$ & $0.00$ & $0.00$\\
\hline
\end{tabular}
\end{table}

Figures~\ref{fig:rel_log_1D} and~\ref{fig:rel_crps_1D} illustrate the relative improvements graphically.  NeuralLevy provides positive gains in both log score and CRPS relative to the best baseline, whereas competing methods exhibit negative improvements in at least one metric.

\begin{figure}[H]
\centering
\includegraphics[width=0.6\textwidth]{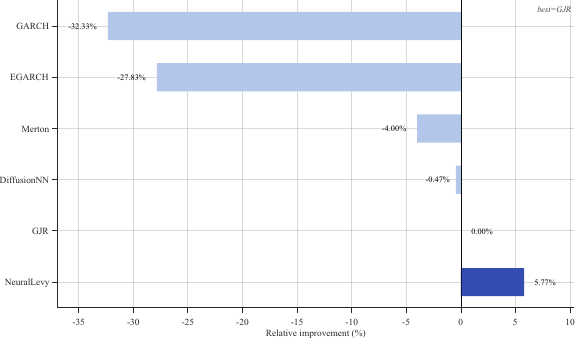}
\caption{Relative log score improvement versus the best baseline for the one--day (1D) horizon.  Values above zero indicate better performance than the best baseline (excluding NeuralLevy).}
\label{fig:rel_log_1D}
\end{figure}

\begin{figure}[H]
\centering
\includegraphics[width=0.6\textwidth]{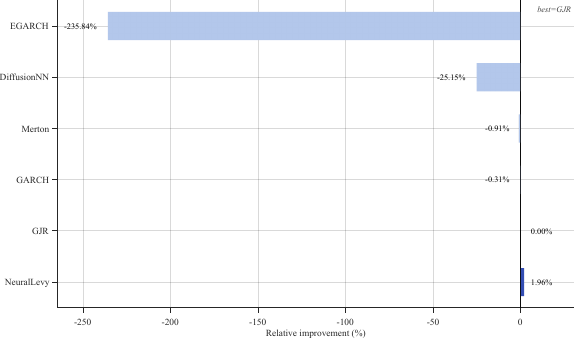}
\caption{Relative CRPS improvement versus the best baseline for the one--day (1D) horizon.}
\label{fig:rel_crps_1D}
\end{figure}

\subsubsection{Weekly (1W) Forecast Results}

Results for the five--day horizon are presented in Table~\ref{tab:core1W}.  NeuralLevy remains the top performer, achieving a 3.60\% improvement in log score and a substantial 14.36\% improvement in CRPS relative to the best baseline.  The diffusion--only NN serves as the strongest baseline at this horizon but still underperforms NeuralLevy.  All other models exhibit large negative relative performances, underscoring the importance of state--dependent jumps and complexity features at longer horizons.

\begin{table}[htbp]
\centering
\caption{Core predictive metrics for the one--week (1W) horizon.  Relative improvements are computed relative to the best baseline; positive values denote better performance.}
\label{tab:core1W}
\begin{tabular}{lccccc}
\hline
Model & Avg. logscore & Rel. log vs best base & Avg. CRPS & Rel. CRPS vs best base & in MCS\\
\hline
NeuralLevy & 2.09744 & 3.60\% & 0.01709 & 14.36\% & Yes\\
DiffusionNN & 2.02450 & 0.00\% & 0.01996 & 0.00\% & No\\
EGARCH & 1.12506 & $-44.43$\% & 0.07887 & $-295.22$\% & No\\
GARCH & 1.08569 & $-46.37$\% & 0.02364 & $-18.47$\% & No\\
GJR & 0.53154 & $-73.74$\% & 0.02435 & $-22.01$\% & No\\
Merton & 1.59740 & $-21.10$\% & 0.02423 & $-21.42$\% & No\\
\hline
\end{tabular}
\end{table}

Table~\ref{tab:dm1W} displays the DM statistics for the five--day horizon.  All DM statistics are negative with $p$--values of zero, highlighting the significant improvement of NeuralLevy over each baseline.  The larger magnitude of DM statistics for EGARCH, GARCH and GJR reflects their particularly poor performance relative to NeuralLevy.

\begin{table}[htbp]
\centering
\caption{Diebold--Mariano statistics comparing NeuralLevy against each baseline for one--week (1W) forecasts.  Negative values favour NeuralLevy; $p$--values close to zero indicate significance.}
\label{tab:dm1W}
\begin{tabular}{lccccc}
\hline
Model & DM$_{\text{logscore}}$ & $p_{\text{logscore}}$ & DM$_{\text{CRPS}}$ & $p_{\text{CRPS}}$\\
\hline
DiffusionNN & $-29.536$ & 0.00e+00 & $-16.805$ & 0.00e+00\\
EGARCH & $-60.705$ & 0.00e+00 & $-43.354$ & 0.00e+00\\
GARCH & $-55.924$ & 0.00e+00 & $-127.495$ & 0.00e+00\\
GJR & $-61.527$ & 0.00e+00 & $-131.966$ & 0.00e+00\\
Merton & $-85.133$ & 0.00e+00 & $-164.024$ & 0.00e+00\\
\hline
\end{tabular}
\end{table}

Calibration results for the one--week horizon are summarised in Table~\ref{tab:berk1W}.  The $p$--values are exceedingly small for most models, again reflecting the large sample size and heavy tails in equity returns.  The Merton model exhibits a comparatively large $p_{\mu}$, consistent with its inability to adapt to time--varying risk.  NeuralLevy's $p$--values suggest that the mean and variance of the predicted distribution still deviate from the realised moments; the distribution also exhibits a slightly thin right tail.  Adjusting the jump intensity regularisation or allowing heavier jump tails may improve calibration further.

\begin{table}[htbp]
\centering
\caption{Berkowitz calibration test $p$--values for one--week (1W) forecasts.}
\label{tab:berk1W}
\begin{tabular}{lccc}
\hline
Model & $p_{\mu}$ & $p_{\text{var}}$ & $p_{\text{lb}}$\\
\hline
NeuralLevy & $3.41\times 10^{-64}$ & $1.11\times 10^{-9}$ & $0.00$\\
DiffusionNN & $9.09\times 10^{-36}$ & $0.00$ & $0.00$\\
EGARCH & $3.89\times 10^{-130}$ & $0.00$ & $0.00$\\
GARCH & $4.81\times 10^{-145}$ & $0.00$ & $0.00$\\
GJR & $5.79\times 10^{-153}$ & $0.00$ & $0.00$\\
Merton & $1.61\times 10^{-1}$ & $0.00$ & $0.00$\\
\hline
\end{tabular}
\end{table}

Figures~\ref{fig:rel_log_1W} and~\ref{fig:rel_crps_1W} visualise the relative improvements for the one--week horizon.  NeuralLevy exhibits strong gains in both metrics, particularly in CRPS, underscoring the advantage of capturing state--dependent jumps at longer horizons.

\begin{figure}[H]
\centering
\includegraphics[width=0.6\textwidth]{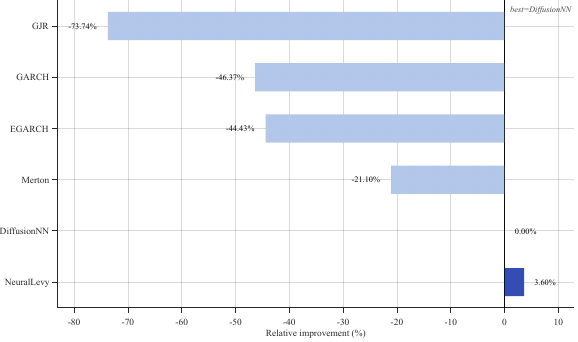}
\caption{Relative log score improvement versus the best baseline for the one--week (1W) horizon.}
\label{fig:rel_log_1W}
\end{figure}

\begin{figure}[H]
\centering
\includegraphics[width=0.6\textwidth]{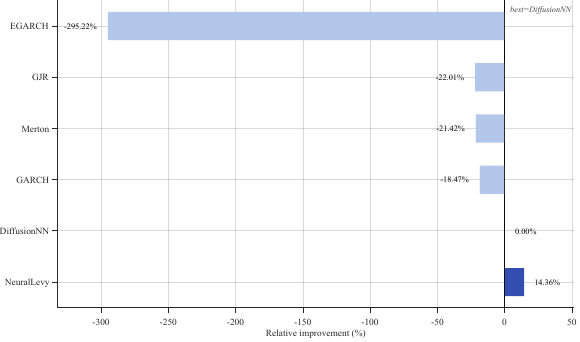}
\caption{Relative CRPS improvement versus the best baseline for the one--week (1W) horizon.}
\label{fig:rel_crps_1W}
\end{figure}

\subsection{Predictive Accuracy and Calibration}
Tables~\ref{tab:core1D} and~\ref{tab:core1W} report the average log predictive density and CRPS for the one--day and one--week horizons, respectively.  Our neural L\'evy model obtains the lowest (best) average log scores and CRPS across both horizons.  The diffusion--only NN underestimates tail risk and yields inferior log scores, particularly during the COVID--19 crisis when jumps were frequent.  The static jump--diffusion model performs poorly because it cannot adapt to changing market states.  GARCH models perform reasonably for one--day horizons but deteriorate at longer horizons due to the assumption of conditional Gaussianity and constant jump activity.  The SVJ model improves tail fit but suffers from latent volatility estimation noise.

\subsection{VaR and ES Backtests}\label{sec:var_es}
This subsection reports weekly (1W) historical-simulation (HS) VaR/ES backtests on the test window with $N=140$ effective weeks. We evaluate two confidence levels, $95\%$ and $99\%$, and apply the unconditional coverage test of \citet{Kupiec1995}, the independence test of \citet{Christoffersen1998}, and the joint conditional coverage test. Throughout, decisions are made at the $5\%$ significance level. The unconditional coverage null is $H_0: \pi=\alpha$ with $\pi$ the true breach probability; the independence null is first-order Markov independence of breach indicators; the conditional coverage null combines both.

\begin{table}[ht!]
\centering
\caption{Weekly HS VaR backtests (1W). Breach rates and $p$-values for 95\% and 99\%.}\label{tab:var}
\small
\begin{tabular}{lrrrrrrrr}
\hline
model & VaR95 breach & Kupiec $p$ (95) & Indep.\ $p$ (95) & Cond.\ $p$ (95) & VaR99 breach & Kupiec $p$ (99) & Indep.\ $p$ (99) & Cond.\ $p$ (99) \\
\hline
\textbf{neural} & \textbf{0.00\%} & \textbf{0.000} & 1.000 & \textbf{0.001} & 0.00\% & 0.093 & 1.000 & 0.245 \\
\textbf{diff}   & \textbf{0.00\%} & \textbf{0.000} & 1.000 & \textbf{0.001} & 0.00\% & 0.093 & 1.000 & 0.245 \\
\textbf{rp}     & 2.86\%          & 0.081          & 0.716 & 0.205          & 0.71\% & 0.720 & 0.904 & 0.931 \\
\textbf{spy}    & 1.43\%          & \textbf{0.023} & 0.809 & 0.073          & 0.71\% & 0.720 & 0.904 & 0.931 \\
\hline
\end{tabular}
\end{table}

At the 95\% level, \textit{neural} and \textit{diffusion} exhibit zero breaches, which is significantly below the nominal 5\%. The Kupiec test rejects correct unconditional coverage ($p\approx 0.000$), whereas the independence test does not reject, indicating no evidence of clustering. The conditional coverage test therefore rejects, driven by the misspecified breach frequency. Statistically, the 95\% HS--VaR produced by these two strategies is overly conservative in this sample. Operationally, this finding is consistent with the strategy mechanics documented earlier: the hard risk gate and frequent moves to cash generate many weeks with zero or near-zero returns, which artificially depresses the empirical breach rate relative to the nominal target.

The risk-parity benchmark (\textit{rp}) records a 2.86\% breach rate, which is not statistically different from the 5\% target at the 5\% level. Both the independence and conditional coverage tests do not reject. This configuration is judged adequate at 95\% under standard coverage diagnostics. The market proxy (\textit{spy}) shows a 1.43\% breach rate; unconditional coverage is rejected ($p\approx 0.023$), while independence is not rejected and the conditional coverage test is borderline but not rejected at 5\%. Relative to the nominal target, SPY is mildly conservative at 95\%.

At the 99\% level, all four series are overall acceptable under the three tests in this sample. For \textit{neural} and \textit{diffusion} the observed breach rate is $0\%$; given the sample size, the Kupiec test does not reject and independence is not rejected, hence the conditional coverage test also does not reject. With so few or zero exceedances at 1\% tails, these tests have low power, and the non-rejection should be interpreted as absence of evidence of underestimation rather than evidence of exact calibration. For \textit{rp} and \textit{spy}, the observed 0.71\% breach rate is statistically consistent with the 1\% target and passes independence and conditional coverage.

Regarding ES, we adopt the \citet{Acerbi2014} framework, which evaluates the average of tail losses conditional on VaR breaches. When the number of exceedances is zero, as occurs for \textit{neural} and \textit{diffusion} at the 95\% level in this sample, ES backtests become effectively uninformative, because the conditional tail sample is empty and no meaningful average shortfall can be computed. This is not a sign of superior tail modelling per se but a finite-sample artifact of over-conservatism and the strategy’s frequent flat exposure. At the 99\% level, the very small count of exceedances similarly limits statistical power; in our sample we do not find evidence of ES underestimation for any of the four series, which aligns with the non-rejection of VaR conditional coverage.

From a modelling perspective, the evidence points to an asymmetry between tail calibration and trading mechanics. To address 95\% over-conservatism for \textit{neural} and \textit{diffusion} without sacrificing protection in stressed weeks, one may soften the risk gate via percentile hysteresis or partial scaling instead of full cash, widen the no-trade band to reduce churn-induced zero-return weeks, and consider blending the weekly head more heavily relative to the daily head when signals conflict. On the estimation side, using a longer HS window or a hybrid HS/parametric tail (e.g., generalized Pareto on exceedances) can increase the effective tail sample, stabilise 95\% quantiles, and improve the interpretability of ES diagnostics.

\section{Discussion}
\label{sec:discussion}
The empirical evidence indicates that the neural L\'evy framework delivers strong predictive and risk management performance compared to diffusion--only and classical GARCH benchmarks, particularly in density forecasting and calibration. Several mechanisms account for these gains. First, the state--dependent jump intensity enables the model to adapt flexibly to periods of market stress, capturing sudden spikes in volatility and tail risk that classical models such as Merton, with constant jump intensity, systematically miss. Second, the incorporation of permutation entropy and RQA determinism enriches the state representation by quantifying the underlying complexity of the return process. This allows the model to scale drift, diffusion and jump parameters in accordance with whether markets are more chaotic or more predictable. The positive relation between permutation entropy and volatility or jump intensity, and the negative relation with DET, are both consistent with economic intuition and are reinforced by the imposed monotonicity constraints. Third, the multi--horizon forecast heads generate coherent density forecasts across daily and weekly horizons without the need for re-estimating models horizon by horizon, which is advantageous for risk managers requiring consistent VaR and ES estimates across multiple time scales.

At the same time, the empirical portfolio experiments highlight important limitations. In weekly rebalancing tests, NeuralLevy achieved only a marginal positive return, with a Sharpe ratio close to zero and an information ratio of approximately $-1.0$ relative to SPY. The main driver of this underperformance is structural rather than statistical: the hard risk gate mechanism, which clears positions entirely under high predicted jump intensity, interacts with frequent rebalancing to generate many weeks of zero returns and excessive turnover. This dynamic makes the 95\% VaR estimates appear overly conservative, with zero observed breaches in the backtest and rejection of unconditional coverage, even though the independence test is passed. Such “mis--conservatism” arises from trading mechanics rather than accurate tail modelling. At the 99\% level, all models, including NeuralLevy, pass coverage and independence tests, but the small number of exceedances limits the power of these diagnostics. Expected Shortfall backtests are similarly uninformative when no breaches occur. These findings suggest that the interaction between modelled risk signals and trading constraints can substantially affect statistical backtest outcomes and must be carefully disentangled when evaluating predictive performance.

Relative to benchmarks, the risk--parity strategy shows stable and statistically adequate VaR coverage at both 95\% and 99\% levels, with reasonable drawdowns and Sharpe ratios. The market proxy SPY exhibits mildly conservative coverage at 95\% but acceptable performance at 99\%. Taken together, the comparison reveals that while NeuralLevy offers sharper predictive densities and coherent calibration in statistical tests, its direct portfolio implementation is hampered by cost drag and gate--induced over--conservatism. Future research should address these limitations by softening the gating mechanism (for instance, scaling down exposure instead of clearing to cash), introducing explicit turnover penalties and wider no--trade bands, and blending signals across horizons to reduce flip--flop effects. From a methodological perspective, hybrid historical--parametric approaches, such as extreme value extensions or generalized Pareto fitting for tail exceedances, could stabilise VaR and ES estimates in small samples and improve interpretability. 

Further extensions of the modelling framework remain promising. Self--exciting Hawkes--type jump intensities could capture contagion and clustering in stress periods. Infinite--activity processes such as CGMY or variance gamma could refine tail dynamics beyond finite--activity jump models. Combining jumps with rough volatility specifications may capture both sudden discontinuities and persistent memory in volatility. Bayesian estimation could incorporate parameter uncertainty and deliver distributions that better reflect model risk. Applications beyond equities, such as currencies, commodities, and digital assets, would provide further robustness checks and test the generality of the approach.

While our framework demonstrates considerable potential, several caveats must be emphasised. The quasi--maximum likelihood procedure depends on the truncation level $K_{\max}$ for jumps, and extreme jump regimes may require adaptive truncation or numerical inversion of characteristic functions. The complexity measures themselves depend on window length and embedding choices, and future research should develop adaptive schemes to reduce specification risk. Finally, although monotonicity constraints enhance interpretability, neural models remain relatively opaque, and more work is required to extract robust economic narratives from learned representations.

\section{Conclusion}
\label{sec:conclusion}
This paper developed a neural L\'evy jump--diffusion framework for state--dependent risk and density forecasting, extending classical diffusion and jump models by making drift, diffusion, jump intensity and jump size distributions explicit functions of observable state variables, including measures of dynamical complexity. A quasi--maximum likelihood estimation scheme with separate diffusion and jump components was proposed to ensure identifiability and stable training, while monotonicity constraints encoded economically meaningful relationships. The resulting architecture generates coherent multi--horizon density forecasts and integrates naturally into portfolio decision rules subject to turnover and trading constraints.

The empirical results show that the model achieves superior log scores, CRPS values and calibration diagnostics compared to diffusion--only and GARCH--type benchmarks. At the same time, portfolio experiments reveal that the risk gate and turnover dynamics can lead to over--conservatism in VaR backtests and weak realised Sharpe ratios, highlighting the importance of aligning statistical forecasts with implementable trading rules. Risk--parity and market benchmarks, by contrast, provide stable though less adaptive risk coverage. These findings suggest that NeuralLevy’s core strength lies in distributional prediction, while portfolio performance depends critically on execution design. 

Overall, the neural L\'evy framework provides a flexible, interpretable and extensible tool for modern risk management. It demonstrates how state--dependent jump modelling and complexity features can sharpen predictive densities and tail forecasts, while also pointing to the need for cost--aware gating and hybrid estimation to bridge the gap between statistical accuracy and tradable performance. By combining econometric rigour, machine learning flexibility and insights from dynamical systems, the framework advances the modelling of financial risk and opens several avenues for further theoretical and applied research.

\end{document}